\documentclass{article}
\usepackage{amsmath, amssymb, amsthm}
\usepackage[margin=1in]{geometry}

\title{Stable non-minimal fixed points of threshold-linear networks}
\author{Jesse Geneson}
\date{}

\newtheorem{remark}{Remark}
\newtheorem{conjecture}{Conjecture}[section]

\newtheorem{theorem}{Theorem}[section]
\newtheorem{proposition}{Proposition}[section]
\newtheorem{corollary}{Corollary}[section]

\begin{document}
\maketitle

\begin{abstract}
In threshold-linear networks (TLNs), a fixed point is called \emph{minimal} if no proper subset of its support is also a fixed point. Curto et al (Advances in Applied Mathematics, 2024) conjectured that every stable fixed point of any TLN must be a minimal fixed point. We provide a counterexample to this conjecture: an explicit competitive TLN on 3 neurons that exhibits a stable fixed point whose support is not minimal (it contains the support of another stable fixed point). We prove that there is no competitive TLN on 2 neurons which contains a stable non-minimal fixed point, so our 3-neuron construction is the smallest such example. By expanding our base example, we show for any positive integers $i, j$ with $i < j-1$ that there exists a competitive TLN with stable fixed point supports $\tau \subsetneq \sigma$ for which $|\tau| = i$ and $|\sigma| = j$. Using a different expansion of our base example, we also show that chains of nested stable fixed points in competitive TLNs can be made arbitrarily long. 
\end{abstract}

\section{Introduction}
Threshold-linear networks (TLNs) are recurrent neural network models in which each neuron’s activity $x_i\ge 0$ evolves according to a linear combination of inputs passed through a rectified linear (ReLU) activation function. TLNs can exhibit multiple stable equilibria (fixed points), which represent memory or activity patterns stored by the network. Curto et al \cite{cgm0} posed the question of whether stable fixed points must be \emph{minimal} in the sense that no smaller subset of neurons could by itself support a (possibly unstable) fixed point. In other words, if a certain set of neurons $\sigma$ is active at a stable equilibrium, is it possible that a strict subset $\tau \subsetneq \sigma$ also yields a fixed point of the system? One might suspect that in \emph{competitive} networks (with inhibitory interactions between neurons), stable equilibria would tend to be minimal, since superfluous active neurons might destabilize the state. 

Threshold-linear networks (TLNs) have become a central model for understanding the combinatorial geometry of neural representations, bridging nonlinear dynamics with network topology. Over the past few decades, a series of works \cite{hsmds, xhs, hss, cdi, cdi0, langdon, langdon0} have established that many qualitative properties of TLNs (existence, multiplicity, and stability of fixed points) can be characterized purely from the signed directed graph of inhibitory and excitatory connections.

Within this framework, a recurring theme has been the search for minimal supports of fixed points, viewed as the fundamental “building blocks” of network memory states. The conjecture that stability implies minimality arose naturally from the observation that in several structured families of TLNs, every stable equilibrium corresponded to a minimal fixed-point support. This suggested a deep principle: that competition among neurons enforces a type of parsimony, preventing redundant active neurons from coexisting at equilibrium.

The present note shows that this principle can fail even in small, fully competitive networks. The counterexample reported here demonstrates that a network may sustain nested stable equilibria, in which a large stable state contains a smaller stable subset. The existence of such stable non-minimal fixed points indicates that competitive dynamics can harbor richer attractor hierarchies than previously recognized. 

In Section~\ref{sec:prelim}, we formally define the TLN model. In Section~\ref{sec:base}, we present a concrete $3$-neuron competitive TLN with a stable fixed point whose support $\{1,2,3\}$ contains a strictly smaller support $\{2\}$ that also supports a stable fixed point. Thus the full-support equilibrium is stable despite not being support-minimal. In Section~\ref{sec:two-neurons}, we show that it is impossible for a non-degenerate competitive TLN on 2 neurons to have a stable non-minimal fixed point support. In Section~\ref{sec:relative}, we expand the 3-neuron example to show for any positive integers $i, j$ with $i < j-1$ that there exists a competitive TLN with stable fixed point supports $\tau \subsetneq \sigma$ for which $|\tau| = i$ and $|\sigma| = j$. In Section~\ref{sec:3k}, we use a different expansion of the 3-neuron example to show that there exist arbitrarily long chains of stable fixed points in competitive TLNs. In Section~\ref{sec:perturb}, we show that these phenomena are robust: small changes to the weight matrix (preserving inhibition) do not destroy the existence or stability of the fixed points or their support relationships. Consequently, there is a continuum of networks exhibiting stable non-minimal fixed points. In Section~\ref{sec:CTLN}, we discuss stable fixed points in combinatorial CTLNs. Finally, we discuss some related directions in Section~\ref{sec:conclusion}.

\section{Preliminaries: TLN Model and Fixed Points}\label{sec:prelim}
We consider a competitive TLN defined by a weight matrix $W \in \mathbb{R}^{n\times n}$ and bias (external drive) vector $b \in \mathbb{R}^n$. The network dynamics are given by 
\[
\dot x = -x + [Wx + b]_+,
\] 
where $[\cdot]_+$ denotes the componentwise ReLU, i.e. $[z]_+ = \max\{z,0\}$ applied to each coordinate of $z$. The term \emph{competitive} means $b_i \ge 0$ and $W_{ii}=0$ for all $i$ (no self-coupling), and $W_{ij} \le 0$ for all $i\neq j$ (all connections between distinct neurons are inhibitory).

A fixed point $x^* \in \mathbb{R}_{\ge 0}^n$ is a nonnegative equilibrium satisfying $x^* = [W x^* + b]_+$. Equivalently, letting $\sigma = \{i : x^*_i > 0\}$ denote the support (active neuron set) of $x^*$, the following conditions must hold \cite{cgm}:
\begin{itemize}
\item \textbf{On-neuron equations:} $(I - W_{\sigma})\,x^*_{\sigma} = b_{\sigma}$, with $x^*_{\sigma} > 0$. Here $W_{\sigma}$ denotes the $|\sigma| \times |\sigma|$ principal submatrix of $W$ on the indices in $\sigma$, and $x^*_{\sigma}, b_{\sigma}$ are the restriction of $x^*, b$ to those coordinates.
\item \textbf{Off-neuron inequalities:} For each $j \notin \sigma$, $b_j + (W_{j,\sigma} \, x^*_{\sigma}) \le 0$. In words, every inactive neuron receives net input at most zero (so that its output remains at 0 under the ReLU).
\end{itemize}
We say that a support $\sigma$ \emph{supports} a fixed point if there exists some $x^* \ge 0$ with $\mathrm{supp}(x^*) = \sigma$ satisfying these conditions. We denote by $\mathrm{FP}(W,b)$ the set of all \emph{supports} of fixed points of $(W,b)$.

A fixed point $x^*$ (with support $\sigma$) is \textbf{stable} if all eigenvalues $\lambda$ of $-I + W_{\sigma}$ have $\Re(\lambda) < 0$ \cite{hss, cdi0}. Equivalently, all eigenvalues $\mu$ of $W_{\sigma}$ satisfy $\Re(\mu) < 1$. We will be especially interested in stable fixed points that are \textbf{non-minimal} in the sense that they have a support $\sigma$ which strictly contains some other support $\tau \subset \sigma$ that also lies in $\mathrm{FP}(W,b)$ (i.e. $\tau$ also supports a fixed point, stable or unstable). 

\section{Example: A 3-Neuron TLN with a Stable Non-Minimal Fixed Point}\label{sec:base}
We first provide an explicit example of a competitive TLN on $n=3$ neurons that possesses a stable equilibrium which is not minimal with respect to fixed point supports. Let the bias be uniform $b = (1,1,1)$, and consider the weight matrix 
\[
W \;=\; \begin{bmatrix}
0 & -\frac{6}{5} & -\frac{4}{5}\\ 
-\frac{3}{2} & 0 & -\frac{3}{5}\\
-\frac{1}{3} & -\frac{3}{2} & 0
\end{bmatrix}.
\] 
In particular $W_{ij}<0$ for all $i\neq j$ and $W_{ii}=0$, so $W$ is competitive. We verify that $(W,b)$ has two stable fixed points supported on $\{1,2,3\}$ and $\{2\}$, respectively, establishing the non-minimal inclusion $\{2\} \subset \{1,2,3\}$.

\subsection*{Support $\{1,2,3\}$ (full support)}
On $\sigma = \{1,2,3\}$, the on-neuron equation is $(I - W)x = b$ (since $\sigma$ includes all neurons). Solving this $3\times 3$ linear system yields 
\[
x^* = \left(\frac{3}{11},\;\frac{5}{11},\;\frac{5}{22}\right),
\] 
where all coordinates are positive. There are no off-neurons in this case, so all fixed point conditions are satisfied and $\{1,2,3\} \in \mathrm{FP}(W,b)$. Next, we examine stability: the Jacobian on this support is $J_{\{1,2,3\}} = -I + W$. The characteristic polynomial is $\det(\lambda I - J_{\{1,2,3\}}) = \lambda^3+3\lambda^2+\frac{1}{30}\lambda+\frac{11}{150}$. All coefficients are positive and $a_2 a_1 - a_3 a_0 = (3)(\frac{1}{30}) - (1)(\frac{11}{150}) = \frac{4}{150} > 0$. Thus, by the cubic Routh-Hurwitz criterion, all eigenvalues have negative real part, hence the full-support fixed point is \textbf{stable}.

\subsection*{Support $\{2\}$ (single neuron)}
Now take $\sigma = \{2\}$. The on-neuron equation $(I - W_{\{2\}})x_{\{2\}} = b_{\{2\}}$ in this case is simply $1\cdot x_2 = 1$, since $W_{22}=0$; thus we get $x_2^* = 1$. For each off-neuron $j \notin \{2\}$, we must check $b_j + W_{j,2} x_2^* \le 0$. Indeed:
\begin{align*}
b_1 + W_{1,2} x_2^* &= 1 + (-\frac{6}{5})(1) \;=\; -\frac{1}{5} \;< 0,\\
b_3 + W_{3,2} x_2^* &= 1 + (-\frac{3}{2})(1) \;=\; -\frac{1}{2} \;< 0.
\end{align*}
Both inequalities hold (in fact strictly). Thus $\{2\}$ satisfies the fixed point conditions, so $\{2\} \in \mathrm{FP}(W,b)$. The Jacobian on this support is a $1\times 1$ matrix $J_{\{2\}} = -1 + W_{22} = -1$. Its single eigenvalue is $-1$ (real part $-1<0$), indicating this fixed point is \textbf{stable} as well.

In particular, we have found two \textbf{stable} fixed points with nested supports: $\{2\} \subset \{1,2,3\}$. The larger support $\{1,2,3\}$ is therefore \emph{not minimal} in $\mathrm{FP}(W,b)$ (since it properly contains $\{2\} \in \mathrm{FP}(W,b)$). This example provides a concrete counterexample to the conjecture from \cite{cgm0} that stable fixed points must be support-minimal in TLNs.

The example is successful because of the careful balance and asymmetry in the inhibitory connections and uniform bias. Neuron 2 has sufficiently strong inhibition onto neurons 1 and 3 to individually suppress them (ensuring ${2}$ is a viable stable state), while the inhibition strengths among all three neurons are not so large as to prevent a full-support equilibrium where all neurons co-exist at reduced activity levels. In particular, neuron 2’s outgoing weights ($W_{1,2}$ and $W_{3,2}$) are larger in magnitude than those from neuron 1 or 3, which allows neuron 2 by itself to drive the net input of neurons 1 and 3 below zero (satisfying the off-neuron conditions for ${2}$). 

At the same time, the full network can settle into a state where each neuron is partially active (since $(I-W)x = b$ has a positive solution), meaning no neuron is completely shut off. This balance yields two alternative stable fixed points: one winner-take-all state with neuron 2 active alone, and one coexistence state with all three neurons active. The key property enabling this non-minimal configuration is that the inhibitory weights are tuned so that a subset of neurons (just neuron 2) can form a stable attractor by itself, yet when all neurons are present they can mutually restrain each other just enough to all remain active and stable. Thus, a larger support can remain stable even though it contains a smaller stable support, violating the support-minimality conjecture.

\section{Minimality for Two Neurons}\label{sec:two-neurons}
Given the example with 3 neurons, it is natural to wonder whether there exists a competitive TLN with only 2 neurons which contains a stable non-minimal fixed point. We show in this section that the answer is no. 

We call a TLN $(W,b)$ \emph{non-degenerate} if $b_i>0$ for at least one $i\in[n]$ and the matrix $I-W_\sigma$ is invertible for every $\sigma\subseteq[n]$.  The requirement $b_i>0$ for some~$i$ rules out the trivial case $b\equiv0$, in which the zero vector (support~$\varnothing$) is a stable fixed point. Without this condition, minimality statements would be vacuous.

We show that $n=3$ is indeed the smallest network size admitting a stable non-minimal fixed point: no non-degenerate competitive TLN on 2 neurons can have a full-support stable fixed point that properly contains another fixed-point support. 

\begin{proposition}\label{prop:n2-minimal}
Let $(W,b)$ be a non-degenerate competitive TLN on $n=2$, with
\[
W=\begin{pmatrix}0 & w_{12}\\ w_{21} & 0\end{pmatrix},\qquad b=(b_1,b_2),\qquad b_i\ge 0,\; w_{12}\le 0,\; w_{21}\le 0.
\]
If $(W,b)$ has a stable fixed point on $\{1,2\}$, then it cannot have a fixed point on either singleton $\{1\}$ or $\{2\}$. In particular, there is no stable non-minimal fixed point for $n=2$.
\end{proposition}

\begin{proof}
Write $w_{12}=-a$, $w_{21}=-d$ with $a,d\ge 0$. The full-support stable fixed point $x=(x_1,x_2)$ solves
\[
(I-W)x=b\quad\Longleftrightarrow\quad
\begin{pmatrix}1 & a\\ d & 1\end{pmatrix}\begin{pmatrix}x_1\\ x_2\end{pmatrix}=\begin{pmatrix}b_1\\ b_2\end{pmatrix}.
\]
Thus
\[
x_1=\frac{\,b_1-a b_2\,}{1-ad},\qquad x_2=\frac{\,b_2-d b_1\,}{1-ad}.
\]
The Jacobian $J=-I+W=\begin{pmatrix}-1 & -a\\ -d & -1\end{pmatrix}$ has eigenvalues $-1\pm\sqrt{ad}$, hence the full-support fixed point is stable whenever $ad<1$, i.e., $1-ad>0$. Existence with $x_1,x_2>0$ forces the strict inequalities
\begin{equation}\label{eq:pos}
b_1>a b_2,\qquad b_2>d b_1.
\end{equation}
(These are exactly $x_1>0$ and $x_2>0$, since $1-ad > 0$.)

Now consider a singleton support, say $\{1\}$. The on-neuron equation gives $x_1=b_1$. The \emph{off}-neuron inequality for neuron $2$ is
\[
b_2+w_{21}x_1=b_2-d b_1\le 0\quad\Longleftrightarrow\quad b_2\le d b_1,
\]
which contradicts the second strict inequality in \eqref{eq:pos}. Hence $\{1\}$ cannot support a fixed point whenever $\{1,2\}$ does.

Similarly, for the singleton $\{2\}$ one needs $x_2=b_2$ and
\[
b_1+w_{12}x_2=b_1-a b_2\le 0\quad\Longleftrightarrow\quad b_1\le a b_2,
\]
which contradicts the first strict inequality in \eqref{eq:pos}. Therefore $\{2\}$ cannot support a fixed point either.

Consequently, in $n=2$ a full-support stable fixed point cannot strictly contain another fixed-point support. This rules out stable non-minimal fixed points for $n=2$.
\end{proof}

\begin{remark}
If $b\equiv 0$, then $x\equiv 0$ (support $\varnothing$) is a fixed point with Jacobian $-I$, hence stable; in this case any other fixed point would not be minimal among supports unless $\varnothing$ is excluded. This motivates the nondegeneracy convention $b_i>0$ for some $i$.
\end{remark}

\section{Relative sizes in pairs of nested stable fixed points}\label{sec:relative}

Given that we found a network with 3 neurons which contains two nested stable fixed points on supports of size $3$ and $1$, it is natural to ask about the combinatorial problem of determining the possible values of $i, j$ with $i < j$ for which there exists a competitive TLN which contains a pair of nested stable fixed points on supports of size $i$ and $j$. In this section, we show that such a competitive TLN exists whenever $i < j-1$.

To see that it is impossible to have a pair of nested stable fixed points on supports of size $j-1$ and $j$, we recall the result from \cite{cgm} that if $\sigma$ is the support of a stable fixed point, then it is impossible to have another stable fixed point on $\sigma - \{k\}$ for any neuron $k \in \sigma$. Thus, we cannot have $i = j-1$.

To prove that there exists a competitive TLN which contains a pair of nested stable fixed points on supports of size $i$ and $j$ for all positive integers $i, j$ with $i < j-1$, we show in the next theorem that there is a simple way to add neurons in any competitive TLN which preserves both fixed points and stability.

\begin{theorem}
\label{thm:neuronsplit}
Consider a competitive TLN $(W,b)$ on $n$ neurons, and suppose $x^*$ is a fixed point supported on $\sigma \subseteq \{1,\dots,n\}$.  Let $i$ be any neuron in $(W,b)$. Construct a new competitive TLN $(W', b')$ by adding $r$ new neurons $i_1,\dots,i_r$ as follows:
\begin{itemize}\itemsep0.5ex
    \item For each original neuron $j$ and each new neuron $i_k$, set $W'_{i_k,\,j} = W_{i,\,j}$.  (Each new neuron $i_k$ receives the same incoming weights from other neurons that $i$ has.)
    \item For each original neuron $p$ and each new neuron $i_k$, set $W'_{p,\,i_k} = 0$.  (Each new neuron exerts no influence on $p$.)
    \item Set all connections between new neurons to zero: $W'_{i_k,\,i_\ell} = 0$ for all $1 \le k,\ell \le r$ (including $k=\ell$, so no self-coupling). All other weights $W'_{ab}$ (when neither $a$ nor $b$ is one of the new neurons) are the same as $W_{ab}$. Also set each new neurons’s bias to the original bias: $b'_{i_k} = b_i$ for $1\le k \le r$, and $b'_j = b_j$ for all other neurons.
\end{itemize}
Then the $(n+r)$-neuron network $(W',b')$ has a fixed point $x'$ which is supported on $\sigma' := \sigma \cup \{i_1,\dots,i_r\}$ if $i \in \sigma$ and supported on $\sigma' = \sigma$ otherwise, given explicitly by 
\[ 
x'_j = 
\begin{cases}
x^*_j, & j \notin \{i_1,\dots,i_r\},\\[6pt]
x^*_i, & j = i_k \text{ for some } k~,
\end{cases} \] 
so that the activity of the original neuron $i$ is replicated in each $i_k$. Moreover, if the original fixed point $x^*$ was stable, then the new fixed point $x'$ is also stable.
\end{theorem}

\begin{proof}
First, we start with the case that $i \in \sigma$. We verify that $x'$ is a fixed point of $(W',b')$ on support $\sigma'$. By construction $x' \ge 0$ and $\mathrm{supp}(x') = \sigma'$. Consider any $i_k$. Its net input in the new network is 
\[ b_i \;+\; \sum_{j \in \sigma} W'_{i_k,\,j}\,x'_j \;+\; \sum_{\ell=1}^r W'_{i_k,\,i_\ell}\,x'_{i_\ell} ~. \] 
Using the definitions of $W'$, this becomes \[b_i + \sum_{j \in \sigma} W_{i,j} x^*_j + \sum_{\ell=1}^r 0 \cdot x^*_i = b_i + \sum_{j \in \sigma} W_{i,j} x^*_j.\] But the original fixed point conditions (on-neuron equation for $i\in\sigma$) give $b_i + \sum_{j \in \sigma} W_{i,j} x^*_j = x^*_i$. Therefore the net input to each $i_k$ is $x^*_i$, and since we set $x'_{i_k} = x^*_i$, it follows that $i_k$ satisfies its on-neuron equation $(I - W'_{\sigma'}) x'_{\sigma'} = b'_{\sigma'}$. 

Now consider any other active neuron $j \in \sigma'$. Its net input under $(W',b')$ is 
\[ b_j \;+\; \sum_{\ell=1}^r W'_{j,\,i_\ell}\,x'_{i_\ell} \;+\; \sum_{p \in \sigma} W'_{j,\,p}\,x'_p~. \] 
Here $W'_{j,\,i_\ell} = 0$ and $x'_{i_\ell} = x^*_i$ for each $\ell$, so the total contribution from the $r$ new neurons is $0$. For the other terms, note that $W'_{j,p} = W_{j,p}$. Hence the net input simplifies to $b_j+W_{j,\sigma} x^*_{\sigma}$. By the original fixed point condition for neuron $j$, this sum equals $x^*_j = x'_j$. Therefore neuron $j$ also satisfies its on-neuron equation in the new network. 

Finally, take any inactive neuron $q \notin \sigma'$ (so $q \notin \sigma$ in the original support as well). Its incoming input in $(W',b')$ is 
\[ b_q \;+\; \sum_{\ell=1}^r W'_{q,\,i_\ell}\,x'_{i_\ell} \;+\; \sum_{p \in \sigma} W'_{q,\,p}\,x'_p~. \] 
For each $\ell$, $W'_{q,\,i_\ell} = 0$, so the total input from all new neurons is $0$. For $p \in \sigma$, $W'_{q,p} = W_{q,p}$. Thus the net input to $q$ becomes $b_q + \sum_{p \in \sigma} W_{q,p} x^*_p = b_q + W_{q,\sigma} x^*_{\sigma}$. Since $q$ was inactive in the original fixed point, the off-neuron inequality gives $b_q + W_{q,\sigma} x^*_{\sigma} \le 0$. Hence $q$ receives net input $\le 0$ in the new network, and remains inactive ($x'_q = 0$). We have now checked all fixed point conditions for $(W',b')$ on support $\sigma'$, so $x'$ is indeed a fixed point of the expanded network.

Next we show that stability is preserved. Let $J_{\sigma} = -I + W_{\sigma}$ be the Jacobian matrix of the original network restricted to the active set $\sigma$ (so all eigenvalues of $J_{\sigma}$ have negative real parts by assumption). We examine the Jacobian $J'_{\sigma'} = -I + W'_{\sigma'}$ of the new network at $x'$ on $\sigma'$. Recall that the eigenvalues of $J'_{\sigma'}$ are the solutions to $\det(J'_{\sigma'}-\lambda I) = 0$, and the matrix $J'_{\sigma'}$ is lower block-triangular if we put the $r$ new neurons in the final $r$ rows and $r$ columns. Thus, the eigenvalues of $J'_{\sigma'}$ consist of the eigenvalues of $J_{\sigma}$, along with $-1$ repeated $r$ times. Therefore, all eigenvalues of $J'_{\sigma'}$ have negative real parts just as in the original network. We conclude that $x'_{\sigma'}$ is a stable fixed point of $(W',b')$. 

Now, we consider the case that $i \not \in \sigma$. In this case, all on-neuron and off-neuron equations are the same as those in $(W, b)$, so $\sigma' = \sigma$ is still a fixed point. Moreover, we have $-I+W_{\sigma'} = -I+W_{\sigma}$, so the eigenvalues are the same. Thus, again we conclude that $x'_{\sigma'}$ is a stable fixed point of $(W',b')$. 
\end{proof}

Let $(W, b)$ denote the 3-neuron competitive TLN with stable fixed points on $\{1,2,3\}$ and $\{2\}$. For any positive integers $i, j$ with $i < j-1$, we split neuron 2 into $i$ neurons $s_1, \dots, s_i$ and we split neuron $1$ into $j-1-i$ neurons $t_1, \dots, t_{j-1-i}$. The resulting competitive TLN has stable fixed points on $\{s_1, \dots, s_i\}$ and $\{t_1, \dots, t_{j-1-i}, s_1, \dots, s_i, 3\}$. Thus, we have derived the following result.

\begin{theorem}
    For any positive integers $i, j$ with $i < j-1$, there exists a competitive TLN with stable fixed point supports $\tau \subsetneq \sigma$ for which $|\tau| = i$ and $|\sigma| = j$.
\end{theorem}

In summary, we have shown that the structured addition of neurons allows precise combinatorial control over the relative sizes of supports in pairs of nested stable fixed points. Importantly, this construction preserves both the fixed point property and stability. It relies on a form of structural redundancy: each added neuron replicates the input profile of an original neuron but does not influence the rest of the network. This flexibility demonstrates a broad failure of support-minimality for stable fixed points in competitive TLNs, thereby opening the door to a richer combinatorial theory of support relationships in competitive TLNs.

\section{Arbitrarily long chains of nested stable fixed points}\label{sec:3k}
Since we found an infinite family of competitive TLNs with pairs of nested stable fixed points, it is natural to investigate how long a chain of nested stable fixed points can be. In this section, we show that such a chain can be arbitrarily long.

We now show how the above $3$-neuron example can be systematically expanded to construct larger TLNs that inherit larger stable non-minimal fixed points. Given the $t\times t$ weight matrix $W$ (with $t = 3$) from the example in Section~\ref{sec:base}, consider for any integer $k \ge 1$ the $t k \times t k$ matrix 
\[ 
W^{(k)} \;=\; W \;\otimes\; I_k, 
\] 
the Kronecker product of $W$ with the $k\times k$ identity matrix $I_k$. Likewise, let $b^{(k)} \in \mathbb{R}^{t k}$ be the all-ones bias vector. It is easy to verify that $W^{(k)}$ is also competitive: $(W^{(k)})_{ii}=0$ and $(W^{(k)})_{ij}\le 0$ for all $i\neq j$. 

To describe the structure of $W^{(k)}$, it is convenient to index the $3k$ neurons as $(i,r)$ where $i\in\{1,2,\dots, t\}$ indicates the neuron’s index within a copy and $r\in\{1,\dots,k\}$ indicates the copy number. By properties of the Kronecker product, we have 
\[
(W^{(k)})_{(i,r),\, (j,s)} \;=\; W_{i,j} \, \delta_{r,s},
\] 
for any $i,j \in \{1,2,\dots, t\}$ and $r,s \in \{1,\dots,k\}$ (here $\delta_{r,s}$ is the Kronecker delta). Thus, $(W^{(k)}, b^{(k)})$ consists of $k$ identical copies of the original $t$-neuron network $(W,b)$, with no coupling between different copies. Since there is no coupling between different copies of $(W,b)$, the stable fixed points of $(W^{(k)}, b^{(k)})$ are simply the disjoint unions of stable fixed points from each copy. 

In particular, starting from the example in Section 3 (where $S=\{1,2,3\}$ and $T=\{2\}$ were both supports of stable equilibria), we obtain that for each $k$, the $3k$-neuron network $(W^{(k)},b^{(k)})$ has a stable full-support equilibrium (on $S^{(k)} = \{1,2,3\}\times\{1,\dots,k\}$, i.e. all $3k$ neurons active) which is \emph{not} minimal, since the strictly smaller support $T^{(k)} = \{(2,r): 1\le r \le k\}$ also supports a stable equilibrium. 

In fact, $(W^{(k)},b^{(k)})$ contain a chain of nested stable fixed points of length $k+1$. Start with $N_1 = T^{(k)} = \{(2,r): 1\le r \le k\}$, and let $N_{i+1}$ be obtained from $N_i$ by adding $(1, i)$ and $(3,i)$ for each $i \le k$. The set $N_i$ supports a stable fixed point for every $i$ with $1 \le i \le k+1$ since $N_i$ is a disjoint union of stable fixed points from each of the $k$ copies. Observe that we have $N_i \subseteq N_j$ for all $i, j$ with $i < j$. Since $k$ can be made arbitrarily large, these chains of nested stable fixed points can be arbitrarily long.

\begin{theorem}
    For all integers $k > 1$, there exists a competitive TLN with a chain $\tau_1 \subsetneq \tau_2 \subsetneq \cdots \subsetneq \tau_k$ of stable fixed point supports. 
\end{theorem}

This construction works because the Kronecker product creates $k$ uncoupled replicas of the base 3-neuron network, each capable of sustaining its own pair of nested equilibria. By taking disjoint unions of those stable supports, we can “stack” these local attractors into a global hierarchy without introducing new cross-inhibitory effects that could destabilize them. The absence of coupling ensures that activating additional copies simply adds independent, already-stable subnetworks, yielding progressively larger stable supports $N_1 \subsetneq N_2 \subsetneq \dots \subsetneq N_{k+1}$. Dynamically, the network behaves like a collection of identical inhibitory motifs that can be turned on one by one without interference, producing an inclusion-ordered chain of equilibria. Conceptually, this reveals that stable non-minimality is not a rare coincidence but a composable property: once a single motif admits nested stable states, arbitrarily long chains can be built by replication. This suggests that the attractor poset of competitive TLNs can, in principle, have unbounded height even under simple weight structures.

\section{Perturbation and robustness of stable fixed point structure}\label{sec:perturb}
The construction in the last section consisted of disjoint copies of the 3-neuron example, with no interaction between the copies (i.e., we had $W_{i j} = 0$ for $i, j$ in different copies). In this section, we show that it is possible to modify the construction so that there is inhibition between \emph{every} pair of neurons, i.e., $W_{i j} < 0$ for all $i, j$. 

In particular, we claim that the existence and stability of the above non-minimal-support equilibria are robust to small perturbations of the weight matrix. Intuitively, since the fixed point conditions are defined by linear equations and inequalities, any solution that satisfies them strictly (with some margin) will persist under sufficiently small changes. The stability criterion (eigenvalues of $-I+W_{\sigma}$) is also an open condition. We formalize this as follows.

\begin{proposition}\label{prop:perturbation}
Let $(W_0,b)$ be a competitive TLN and let $\sigma \subseteq \{1,\dots,n\}$ be the support of a fixed point $x^0 \in \mathbb{R}^n_{\ge0}$. Assume that:
\begin{enumerate}
    \item $x^0_{\sigma} > 0$ and $(I - (W_0)_{\sigma})\, x^0_{\sigma} = b_{\sigma}$ (so $x^0$ satisfies the on-neuron equations on $\sigma$),
    \item $b_j + (W_0)_{j,\sigma} x^0_{\sigma} < 0$ for every $j \notin \sigma$ (strict off-neuron inequalities), and 
    \item all eigenvalues of $-I + (W_0)_{\sigma}$ have negative real parts (so the fixed point on $\sigma$ is stable in $(W_0,b)$).
\end{enumerate}
Then there exists $\varepsilon > 0$ such that for any weight matrix $W$ (of the same size) satisfying $\|W - W_0\| < \varepsilon$ and $W_{ii}=0$, $W_{ij}<0$ for $i\neq j$ (i.e. $W$ is strictly competitive and is a small perturbation of $W_0$), the TLN $(W,b)$ has a fixed point $x$ with support $\sigma$ that is also stable. Moreover, $x$ depends continuously on $W$ and for $W$ close to $W_0$ we have $x_{\sigma}>0$ and the same strict inequalities $b_j + W_{j,\sigma} x_{\sigma} < 0$ for $j \notin \sigma$.
\end{proposition}

\begin{proof}
Consider the system of equations given by the on-neuron conditions: $F(W_{\sigma}, x_{\sigma}) := (I - W_{\sigma}) x_{\sigma} - b_{\sigma} = 0$. By assumption, at $W_0$ there is a solution $x^0_{\sigma}$ of $F(W_{0,\sigma}, x_{\sigma})=0$ with $x^0_{\sigma} > 0$. Furthermore, the matrix $\partial F/\partial x_{\sigma} = I - (W_0)_{\sigma}$ is invertible (since condition (3) implies $\det(I - (W_0)_{\sigma}) \neq 0$. Therefore, by the implicit function theorem, there exists a continuously differentiable function $x_{\sigma}(W_{\sigma})$ defined in a neighborhood of $W_{0,\sigma}$ such that $F(W_{\sigma}, x_{\sigma}(W_{\sigma}))=0$ and $x_{\sigma}(W_{0,\sigma}) = x^0_{\sigma}$. In other words, for all sufficiently small perturbations of the submatrix $W_{\sigma}$, we obtain a corresponding solution $x_{\sigma}$ (close to $x^0_{\sigma}$) of the on-neuron equations. By continuity, $x_{\sigma}(W_{\sigma})$ will remain strictly positive for $W$ sufficiently close to $W_0$ (since $x^0_{\sigma} > 0$). Lifting this to the full $n$-dimensional state, we set $x_i = 0$ for $i \notin \sigma$ to obtain a candidate fixed point $x(W)$ of $(W,b)$ on support $\sigma$. 

We next check the off-neuron inequalities and stability for the perturbed system. For each $j \notin \sigma$, define $\Phi_j(W) := b_j + W_{j,\sigma} x_{\sigma}(W_{\sigma})$. At $W=W_0$, we have $\Phi_j(W_0) = b_j + (W_0)_{j,\sigma} x^0_{\sigma} < 0$ by assumption (2). Since $\Phi_j$ depends continuously on $W$, there is a neighborhood of $W_0$ in which $\Phi_j(W)$ remains strictly negative for all $j \notin \sigma$. Thus for $W$ sufficiently close to $W_0$, all off-neuron conditions $b_j + W_{j,\sigma} x_{\sigma}(W_{\sigma}) \le 0$ hold (in fact with strict inequality), meaning $\sigma$ indeed supports a fixed point $x(W)$ of $(W,b)$ with the desired property $x_i(W)>0$ iff $i\in\sigma$. 

Finally, consider the stability condition. Write the perturbed Jacobian on support $\sigma$ as $J_{\sigma}(W) = -I + W_{\sigma}$. When $W=W_0$, $J_{\sigma}(W_0)$ has eigenvalues with negative real parts by assumption (3). The eigenvalues of $J_{\sigma}(W)$ depend continuously on $W_{\sigma}$, since the roots of any polynomial depend continuously on its coefficients. Hence there exists $\varepsilon' > 0$ such that if $\|W_{\sigma} - (W_0)_{\sigma}\| < \varepsilon'$, then every eigenvalue $\lambda$ of $J_{\sigma}(W)$ satisfies $\Re(\lambda) < 0$ (since none of the eigenvalues of $J_{\sigma}(W_0)$ are on the imaginary axis, a sufficiently small perturbation cannot cross the stability boundary). This ensures that the fixed point $x(W)$ is linearly stable for $W$ in a neighborhood of $W_0$. We may choose $\varepsilon>0$ small enough that the conditions for both the existence of $x(W)$ and stability of $J_{\sigma}(W)$ are simultaneously satisfied, completing the proof.
\end{proof}

This perturbation result implies that the special structure of the Kronecker product expansion is \emph{not} essential for the existence of the stable non-minimal fixed point patterns. Any sufficiently small deformation of $W^{(k)}$ will preserve the key features. In particular, we have:

\begin{corollary}\label{cor:uncountable}
For each $k \ge 3$, there is an open set $U$ of strictly competitive $k \times k$ weight matrices such that every $W \in U$ (with bias $b^{(k)}$) has a stable fixed point whose support is not minimal.
\end{corollary}

In summary, starting from one concrete $3$-neuron counterexample to minimality, we have shown how to generate large families of higher-dimensional examples. The Kronecker product construction produces structured large networks with stable non-minimal equilibria, and the perturbation argument shows that this phenomenon persists beyond that specific structured form, occupying an open set in the space of all competitive weight matrices. 

\section{Combinatorial TLNs}\label{sec:CTLN}

A \emph{combinatorial threshold--linear network} (CTLN) is a special subclass of TLNs specified by a directed graph $G$ on $n$ vertices and parameters $(\varepsilon,\delta,\theta)$ \cite{curto, cm, cm0, parmeleeetal, parmelee}.
The weight matrix $W=W(G;\varepsilon,\delta)$ and bias $b$ are
\begin{equation*}
W_{ii}=0,\qquad 
W_{ij}=
\begin{cases}
-1+\varepsilon,& \text{if } j\to i \text{ in }G,\\[4pt]
-1-\delta,& \text{if } j\not\to i,
\end{cases}
\quad (i\neq j),
\qquad
b_i=\theta>0~.
\end{equation*}
We work in the standard \emph{legal range} for CTLNs:
\begin{equation*}
\delta>0,\quad 0<\varepsilon<1,\quad \text{and}\quad \varepsilon<\frac{\delta}{\delta+1}.
\end{equation*}
Thus all off--diagonal entries are inhibitory, with edges of $G$ corresponding to \emph{weaker} inhibition $-1+\varepsilon$ and non--edges to \emph{stronger} inhibition $-1-\delta$.  The dynamics are
$\dot x=-x+[Wx+b]_+$, so fixed points are determined by the usual on--neuron equations and off--neuron inequalities.

\paragraph{Clique supports and target--free cliques.}
A subset $\sigma\subseteq[n]$ is a (directed) \emph{clique} if $G|_\sigma$ has all bidirectional edges between distinct vertices.
A vertex $v\notin\sigma$ is a \emph{target} of $\sigma$ if every $j\in\sigma$ has $j\to v$ in $G$.
A clique $\sigma$ is \emph{target--free} if it has no external target.
In the legal range, it has been shown that every target-free clique supports a stable fixed point, and it is conjectured that these are the only stable fixed points in combinatorial TLNs.

\begin{conjecture}\label{conj:target}\cite{div, cgm0}
    Let $W = W(G,\varepsilon,\delta)$ be a CTLN on $n$ nodes with graph $G$, let $\sigma \subseteq [n]$, and suppose that $\delta > 0$ and $0 < \varepsilon < \delta/(\delta+1)$. Then $\sigma$ is the support of a stable fixed point if and only if $\sigma$ is a target-free clique.
\end{conjecture}

Conjecture~\ref{conj:target} has been proved in some special cases, including for oriented graphs \cite{div} and symmetric graphs \cite{cm0}. In \cite{cgm0}, the conjecture was proved for all combinatorial TLNs of size $n \le 4$ with $\delta > 0$ and $0 < \varepsilon < \delta/(\delta+1)$, as well as all combinatorial TLNs of any size with $0<\varepsilon<1$ and \[\delta > \frac{\varepsilon}{1-\varepsilon}(n^2-n-1).\]

\paragraph{Minimality of clique--supported fixed points.}
For completeness, we include a short proof that \emph{any} CTLN fixed point supported on a clique is automatically minimal among fixed--point supports.  In particular, this yields that every \emph{stable} clique--supported fixed point is minimal.

\begin{proposition}\label{prop:ctln-clique-minimal}
Let $(W,b)$ be a CTLN in the legal range, and suppose that $\sigma$ is a clique. Then no proper subset $\tau\subsetneq\sigma$ supports a fixed point. 
\end{proposition}

\begin{proof}
Curto et al \cite{cgm} showed that all clique fixed points must be target-free, so $\tau$ cannot be a fixed point since every element of $\sigma - \tau$ is a target.
\end{proof}

If Conjecture~\ref{conj:target} is true, then the last lemma would imply that every stable fixed point support in a combinatorial CTLN is minimal if $\delta > 0$ and $0 < \varepsilon < \delta/(\delta+1)$.

\paragraph{Connection to the minimality conjecture.}
Curto et al \cite{cgm0} conjectured that in \emph{any} TLN, every stable fixed point is minimal. We showed this is false in general, but Proposition~\ref{prop:ctln-clique-minimal} shows that it holds for CTLNs whenever the stable fixed point is supported on a clique.  Thus, within the combinatorial TLN family, the minimality conjecture is confirmed for the ubiquitous class of clique--supported attractors.

\section{Conclusion}\label{sec:conclusion}
We have demonstrated that stable fixed points in competitive TLNs need not be minimal. In particular, we presented an explicit $3$-neuron network with a stable equilibrium whose full support (all three neurons) contains a strictly smaller support (a single neuron) that is itself a stable equilibrium. Using this counterexample as a building block, we showed for any positive integers $i, j$ with $i < j-1$ that there exists a competitive TLN with stable fixed point supports $\tau \subsetneq \sigma$ for which $|\tau| = i$ and $|\sigma| = j$. We also constructed $3k$-neuron networks (for arbitrary $k$) via a Kronecker product with the identity, yielding chains of nested stable fixed points of arbitrary length. Finally, we showed that these patterns are robust to small inhibitory perturbations of the weight matrix, implying the existence of a continuum of competitive TLNs that exhibit such stable non-minimal equilibria. Together, these findings refute the conjecture of Curto et al.~\cite{cgm0} that any stable fixed point must have a minimal support, and reveal that even purely competitive (inhibition-dominated) networks can sustain coexisting stable states with nested supports. In other words, competitive TLNs can possess a richer and less intuitive attractor landscape than previously assumed.

These discoveries open up several directions for further investigation. First, an important question is how to classify the minimal network motifs that admit stable non-minimal fixed points. Our $3$-neuron example shows that three neurons suffice in the competitive setting (while two neurons do not), establishing one minimal motif for this phenomenon. However, stable non-minimal equilibria are not exclusive to this $3$-neuron pattern; indeed, distinct network architectures can also produce stable states with nested supports. Identifying the full spectrum of such motifs and support-inclusion configurations remains a key direction for future work.

Second, the presence of nested stable states suggests the possibility of a \emph{hierarchy of attractors} in inhibitory networks: stable equilibria might form inclusion chains $\tau_1 \subsetneq \tau_2 \subsetneq \cdots \subsetneq \tau_k$, representing multiple layers of memory. We showed that such chains of stable fixed points can be made arbitrarily long in competitive TLNs. Determining what connectivity constraints might limit or enable such chains is an intriguing question for future work. 

A natural structural object arising from these results is the poset of stable fixed point supports, ordered by inclusion. Each node in this poset corresponds to a stable equilibrium, and the covering relations encode minimal transitions between coexisting attractors. For networks with multiple nested equilibria, this poset captures the combinatorial architecture of the attractor landscape. For instance, in our examples the poset contains chains of arbitrary length, but more generally it may contain branching and incomparable elements. Analyzing the topology of this poset, e.g., its height, width, connected components, or lattice properties, could reveal how inhibitory connectivity constraints shape the overall organization of stable states in competitive TLNs.

Finally, from a dynamical perspective, it is natural to ask how trajectories in a TLN select or transition between these nested equilibria, especially in the presence of perturbations or noise. For example, do transitions from a smaller-support equilibrium to a larger-support equilibrium (or vice versa) correspond to specific bifurcations as network parameters vary? Can features of the network’s connectivity (such as particular inhibitory motifs or feedback loops) predict the existence of stable non-minimal states? Addressing these issues would help advance a broader theoretical understanding connecting the combinatorial structure of $W$ with the topology of the attractor landscape in threshold-linear systems.

For the restricted family of combinatorial TLNs (CTLNs), our results highlight a fundamental divergence. The target-free clique conjecture in CTLN theory posits that every stable equilibrium’s support must correspond to a target-free clique in the network’s directed graph. If this conjecture holds, then no stable state in a CTLN can have a support that strictly contains another smaller fixed point support; in other words, all stable fixed points in such networks would necessarily be support-minimal. In stark contrast, we have shown that general competitive TLNs (unconstrained by the CTLN architecture) can indeed sustain coexisting stable states with nested supports. This means that purely inhibitory networks, when not limited to the combinatorial regime, are capable of the very hierarchical attractor relationships that CTLNs are conjectured to forbid, revealing a richer repertoire of steady states than previously recognized.

From a broader perspective, the ability of competitive TLNs to exhibit nested stable equilibria underscores the greater richness of their attractor structure compared to that of CTLNs. Whereas CTLNs under the target-free clique conjecture enforce a single-layer, minimal-support pattern for stable fixed points, general competitive TLNs can support hierarchies of attractors: sequences of stable equilibria embedded one within another. This suggests that the dynamical repertoire of TLNs is significantly more complex than what the combinatorial framework had captured. In summary, our findings broaden the understanding of threshold-linear network dynamics, indicating that when freed from strict structural constraints, even simple competitive TLNs can harbor multi-level stable states and far more intricate attractor landscapes than previously appreciated.

\section*{Acknowledgments}

We used GPT-5 to find the $3$-neuron counterexample and to draft some of the expository text.

\end{document}